\documentclass{article}
\usepackage{amsthm}
\author{Jean-Guillaume Dumas\thanks{
    Laboratoire J. Kuntzmann,
    Universit\'e de Grenoble. 51, rue des Math\'ematiques, umr CNRS
    5224, bp 53X, F38041 Grenoble, France,
    \href{mailto:Jean-Guillaume.Dumas@imag.fr}{Jean-Guillaume.Dumas@imag.fr},
    \href{http://ljk.imag.fr/membres/Jean-Guillaume.Dumas/}{ljk.imag.fr/membres/Jean-Guillaume.Dumas}.
  }
  \and Erich Kaltofen\thanks{
    Department of Mathematics.
    North Carolina State University.
    Raleigh, NC 27695-8205, USA.
    \href{mailto:kaltofen@math.ncsu.edu}{kaltofen@math.ncsu.edu},
    \href{http://www.kaltofen.us/}{www.kaltofen.us}.
  }
  \and Emmanuel Thom\'e\thanks{
    Caramba -- INRIA Nancy Grand Est.
    615, rue du Jardin Botanique -- 54600 Villers-l\`es-Nancy -- France.
    \href{mailto:Emmanuel.Thome@inria.fr}{Emmanuel.Thome@inria.fr},
    \href{http://www.loria.fr/~thome/}{www.loria.fr/\~{}thome/}.
  }
  \and Gilles Villard\thanks{
    Laboratoire LIP   [CNRS  ENSL  INRIA  UCBL  U. Lyon ] 
  46, All\'ee d'Italie F69364   Lyon   Cedex 07 -- France.
    \href{mailto:gilles.villard@ens-lyon.fr}{gilles.villard@ens-lyon.fr},
    \href{http://perso.ens-lyon.fr/gilles.villard/}{perso.ens-lyon.fr/gilles.villard}.
  }
}
\newenvironment{multlineeq}{\equation}{\endequation}
\newcommand{\newlineeq}{}

\newcount\addauflag
\global\addauflag=4
\def\addauthorsection{}

\newcommand{\swldots}{\ensuremath{\ldots}}
\newcommand{\swcdots}{\ensuremath{\cdot}\ensuremath{\cdot}\ensuremath{\cdot}}
\newcommand{\swle}{\ensuremath{\le}}

\newcommand{\swge}{\ensuremath{\ge}}
\newcommand{\swne}{\ensuremath{\ne}}
\newcommand{\sweq}{\ensuremath{=}}

\newcommand{\swminus}{\ensuremath{-}}
\newcommand{\swplus}{\ensuremath{+}}
\newcommand{\swin}{\ensuremath{\in}}

\newcommand{\swchecks}{\ensuremath{\checks}}

\title{Linear Time Interactive Certificates for the Minimal Polynomial and the Determinant of a Sparse Matrix\footnote{This material is based on work
      supported in part by the Agence Nationale pour la Recherche
      under Grant ANR-11-BS02-013 HPAC and the \href{http://opendreamkit.org}{OpenDreamKit} \href{https://ec.europa.eu/programmes/horizon2020/}{Horizon 2020} \href{https://ec.europa.eu/programmes/horizon2020/en/h2020-section/european-research-infrastructures-including-e-infrastructures}{European Research Infrastructures} project (\#\href{http://cordis.europa.eu/project/rcn/198334_en.html}{676541}).}}

\usepackage{amsmath,amssymb,amsfonts,xspace,graphicx,multirow}
\usepackage[all]{xy}

\usepackage{enumerate}

\newcommand{\rank}{\mathop{\text{\upshape rank}}\nolimits}

\newcommand{\Z}{\ensuremath{{\mathbb Z}}}
\newcommand{\F}{\ensuremath{{\mathbb F}}}
\newcommand{\KK}{\F}
\newcommand{\ZZ}{\Z}

\usepackage{algorithm,algorithmic}
\newcommand{\checks}{\ensuremath{\mathrel{\stackrel{?}{=\!=}}}}
\newcommand{\REDUC}[2]{{\sc #1}$\prec${\sc #2}}

\newtheorem{theorem}{Theorem}
\newtheorem{corollary}[theorem]{Corollary}
\newtheorem{lemma}[theorem]{Lemma}
\newtheorem{definition}[theorem]{Definition}
\newtheorem{remark}[theorem]{Remark}
\newtheorem{proposition}[theorem]{Proposition}

\usepackage{color,svgcolor}
\usepackage[plainpages=true,bookmarks=false]{hyperref}
\makeatletter
\hypersetup{
pdftitle={\@title},
pdfauthor={Dumas,Kaltofen,Thom\'e,Villard},
breaklinks=true,
colorlinks=true,
linkcolor=darkred,
citecolor=blue,
urlcolor=darkgreen,
}
\makeatother

\newcommand{\overlongrightarrow}[2]
{\vbox{\offinterlineskip
      \halign{##\cr
              \hfill #1 \hfill\cr
              \hbox to #2{\relax\rightarrowfill}\cr}}\ignorespaces
}

\newcommand{\overlongleftarrow}[2]
{\vbox{\offinterlineskip
      \halign{##\cr
              \hfill #1 \hfill\cr
              \hbox to #2{\relax\leftarrowfill}\cr}}\ignorespaces
}

\begin{document}
\maketitle
\begin{abstract}
Computational problem  certificates %
are
additional data structures for each output, which can be used
by a---possibly randomized---verification algorithm that proves the
correctness of each output. 
In this paper, we give an algorithm that computes  %
a certificate 
for the minimal polynomial of
sparse or structured $n\times n$ matrices over an abstract field, of
sufficiently large cardinality, 
whose Monte Carlo verification complexity requires a single matrix-vector
multiplication and a linear number of extra field operations. 
We also propose a novel preconditioner that ensures irreducibility of the
characteristic polynomial of the generically %
preconditioned matrix. 
This preconditioner takes linear time to be applied and uses only two
random entries. 
We then combine these two techniques to give algorithms that compute
certificates for the determinant, and thus %
 for the characteristic polynomial,
whose Monte Carlo verification complexity is therefore also linear. 
\end{abstract}

\section{Introduction}
We consider a square sparse or structured matrix $A\in\F^{n\times{}n}$,
where $\F$ is an exact field.
By sparse or structured we mean that multiplying a vector by $A$
requires fewer operations than a dense matrix-vector multiplication. 
The arithmetic cost to apply $A$ is denoted by $\mu(A)$ which thus satisfies
$\mu(A)\leq{}n(2n-1)$ ($n^2$ multiplications and $n(n-1)$ additions). 

For such sparse matrices, Wiedemann's algorithm,
together with some preconditioning~\cite{CEKSTV02}, provides means to compute
the rank, the minimal polynomial or 
the characteristic polynomial  %
of sparse
matrices, via the computation of the minimal polynomial of its associated
sequence projected at random values~\cite{Wiedemann:1986:SSLE}.\\\ \\[15pt]

The novelty of this paper is to provide an algorithm that computes a certificate
for the minimal polynomial of sparse 
matrices.
The Monte Carlo verification complexity of our certificate is linear in the input size
and this certificate is composed of the minimal polynomial itself, of three other
polynomials and of a vector. 
The verification procedure used throughout this paper is that of {\em interactive certificates} with the taxonomy
of~\cite{jgd:2014:interactivecert}.  
Indeed, we consider a {\em Prover}, nicknamed {\em Peggy}, who
will perform a computation and provide additional data
structures. %
We also consider a {\em Verifier}, nicknamed {\em Victor}, who will
check the validity of the result, %
faster than by just recomputing it. 
By {\em certificates} for a problem that is given by input/output
specifications, we mean, as in \cite{KLYZ09,Kaltofen:2011:quadcert}, 
\emph{an input-dependent data structure and an algorithm that computes from that input
and its certificate
the specified output, and that has lower computational
complexity than any known algorithm that does the same when only receiving the
input. Correctness of the data structure is not assumed but validated by the
algorithm.} 
By {\em interactive certificate}, we mean interactive proofs, similar to %
$\sum$-protocols (as in~\cite{Cramer:1997:Sigma}) were  
the Prover submits 
a {\em Commitment}, that is some result of a computation; 
the Verifier answers by
a {\em Challenge}, usually some uniformly sampled random values;
the Prover then answers with 
a {\em Response}, that the Verifier can use to convince himself of the validity
of the commitment. Several {\em rounds} of challenge/response might be necessary
for Victor to be fully convinced.
Such proof systems are said to be {\em complete} if the probability
that a true statement is rejected by the Verifier can be made arbitrarily
small; and {\em sound} if the probability that a false
statement is accepted by the Verifier can be made arbitrarily small.
In practice it is sufficient that those probability are $<1$, as the protocols
can always be run several times.
Some of our certificates will also be {\em perfectly
  complete}, that is a true statement is never rejected by the Verifier.

All our certificates can be simulated non-interactively by
Fiat-Shamir heuristic~\cite{Fiat:1986:Shamir}: uniformly sampled random values
produced by Victor are replaced by hashes of the input and of previous
messages in the protocol.  
Complexities are preserved, as producing cryptographically strong
pseudo-random bits by a cryptographic hash function (e.g., like the
extendable output functions of the SHA-3 family defined
in~\cite{Bertoni:2010:sponge,SHAKE}), is linear in the size of both its input and
output. 
More precisely, we will need $O(n\log(n))$ cryptographically strong random bits
for the minimal polynomial certificate of Figure~\ref{fig:MinPoly}, and only
$O(\log(n))$ for the determinant,
Figures~\ref{fig:WiedDet} and~\ref{fig:GammaDet}.

There may be two main ways to design such certificates.
On the one hand, efficient protocols can be designed for delegating
computational tasks. In recent years, generic protocols have been designed for
circuits with polylogarithmic
depth~\cite{Goldwasser:2008:delegating,Thaler:2013:crypto}.   
The resulting protocols are interactive and their cost for the Verifier is
usually only roughly proportional to the input size. They however can
produce a non negligible overhead for the Prover and are restricted to
certain classes of circuits. 
Variants with an amortized cost for the Verifier can also be designed,
see for instance~\cite{Parno:2013:Pinocchio}, quite often using relatively
costly homomorphic routines.
We here however want the Verifier to run faster than the Prover, so we discard
amortized models where the Verifier is allowed to do a large amount of
precomputations, that can be amortized only if, say, the same matrix is
repeatedly used~\cite{Chung:2010:delfhe,Gentry:2014:nizkfhe}.

On the other hand, dedicated certificates (data structures and
algorithms that are publicly verifiable a posteriori, without interaction) have
also been developed, %
e.g., for dense exact linear
algebra~\cite{Freivalds:1979:certif,Kaltofen:2011:quadcert,Fiore:2012:PVD}.
There the certificate constitutes a proof of correctness of a result,
not of a computation, and can thus also stand an
independent, computation error-correcting
verification.
The obtained certificates are problem-specific,
but try to reduce as much as possible the
overhead for the Prover, while preserving a fast verification procedure.
In the current paper %
we give new problem-specific
certificate with fast verification and negligible overhead for the Prover.

In exact linear algebra, the  simplest problem with an optimal certificate
is the linear system solution, {\sc LinSolve}: for a matrix~$A$ and a
vector~$b$, checking that $x$ is actually a solution is done by  one
multiplication 
of~$x$ by~$A$. The cost of this check is similar to that of just enumerating all
the non-zero coefficients of~$A$. Thus certifying a linear system is reduced to
multiplying a matrix by a vector: \REDUC{LinSolve}{MatVecMult}.
More precisely, by \REDUC{A}{B}, we mean that there exists certificates for {\sc
  A} that use certificates for {\sc B} whose verification times are
  essentially similar: $\operatorname{Verif}(A)=\operatorname{Verif}(B)^{1+o(1)}$.
In~\cite{jgd:2014:interactivecert}, two reductions have been
made: first, that the rank can be certified via certificates for linear systems;
second, that the characteristic polynomial can be certified via certificates for
the determinant: \REDUC{CharPoly}{Det} and \REDUC{Rank}{LinSolve}.
The verification procedure for the rank
is essentially optimal, 
it requires two matrix-vector products and $n^{1+o(1)}$ additional operations;
while the verification of the characteristic polynomial after
  verification of a determinant is simply linear. 
No reduction, however, was given for the determinant. 
We bridge this gap in this paper.
Indeed, we show here that the computation of the minimal polynomial can be
checked in linear time by a single matrix-vector multiplication:
\REDUC{MinPoly}{MatVecMult}. 
Then we use Wiedemann's reduction of the determinant
to the minimal polynomial,
\REDUC{Det}{MinPoly},~\cite{Wiedemann:1986:SSLE,KaPa91}, and propose a
more efficient preconditioning for the same reduction.

This paper comes with a companion paper~\cite{jgd:2016:intdet} that solve
similar problems but with different techniques. We nonetheless believe that they
are of independent interest, as shown by the following comparison of their
salient differences: that  %

\newsavebox{\bulspace}\savebox{\bulspace}{\hspace*{0em}$\bullet$\hskip\labelsep}%
{\setlength{\leftmargini}{\wd\bulspace}%
\begin{enumerate}
\item[$\bullet$]
The paper~\cite{jgd:2016:intdet} gives certificates for the
  Wiedemann sequence, while we here directly certify its minimal
  polynomial; \\[-0.5cm]
\item[$\bullet$]
Complexities for the Verifier time and the extra 
  communications are linear here while they are increased by $(\log
  n)^{\Omega(1)}$ in~\cite{jgd:2016:intdet};  \\[-0.5cm]
\item[$\bullet$]
The verification in~\cite{jgd:2016:intdet} requires a black box for the
  transposed matrix; \\[-0.5cm]
\item[$\bullet$]
The certificates here are interactive 
  while in~\cite{jgd:2016:intdet}, one of the certificates is, up to our
  knowledge, the only known %
non-interactive protocol for the determinant with Prover complexity $n^{1.5+o(1)}$. 
\end{enumerate}
}%

The paper is organized as follows.
We first present in Section~\ref{sec:simplecert} a new multiplicative
preconditioner that allows to check the determinant as a quotient of minors.
In Section~\ref{sec:wiedseq} we define Wiedemann's projected Krylov sequence and
propose a Monte Carlo certificate for the minimal polynomial of this sequence in
Section~\ref{sec:fAuvcert}. 
We apply this with random projections in Section~\ref{sec:minpol}, 
which provides a certificate for the
minimal polynomial of the matrix. In Section~\ref{sec:diag}, we see that with a
diagonal preconditioning, we obtain another certificate for the determinant.
In Section~\ref{sec:gammadet}, we then combine this idea with the preconditioner of
Section~\ref{sec:simplecert}  to obtain a
more efficient certificate for the determinant. 
This can be combined with the
characteristic polynomial reduction of~\cite{jgd:2014:interactivecert}, in order
to provide also a linear time certificate for the characteristic polynomial of
sparse or structured matrices.

\section{A simple interactive certificate for determinant}
\label{sec:simplecert}

We first present a new multiplicative preconditioner that enables the Prover
to compute our simple certificate,
which is based on the characteristic matrix
of the companion matrix of the polynomial $z^n+\sigma$ (see (\ref{eq:Gamma}) below).

\begin{lemma}\label{lem:Gamma}
\begin{equation}\label{eq:Gamma}
\text{Let}~
\Gamma(\sigma,\tau) =
\begin{bmatrix}
\tau   & -1   &    0   &\ldots & 0
\\
0     & \tau  &  -1    &\ddots & \vdots
\\
\vdots&   0  & \ddots &\ddots & 0
\\
0     &      & \ddots & \tau   & -1
\\
\sigma & 0   & \ldots & 0     & \tau
\end{bmatrix} \in \KK[\sigma,\tau]^{n\times n},
\end{equation}
where $\tau$ and $\sigma$ are variables.
If $A\in\KK^{n\times n}$ is non-singular,
then 
$\det(\lambda I_n - A\>\Gamma(\sigma,\tau))$ is
irreducible in $\KK(\sigma,\tau)[\lambda]$.
\end{lemma}

\noindent{\sc Proof.}
We observe that
for $\Gamma(\sigma,\tau)$ in (\ref{eq:Gamma}) we have $\det(\Gamma(\sigma,\tau)) = \tau^n + \sigma$,
which is an irreducible
polynomial in the bivariate polynomial domain $\KK[\sigma,\tau]$.
Next we consider the characteristic polynomial of $B(\sigma,\tau) = A\>\Gamma(\sigma,\tau)$,
namely
\begin{multlineeq}\label{eq:cBmu}
c^{B(\sigma,\tau)}(\lambda) = \det(\lambda I_n - B(\sigma,\tau))\newlineeq
= \lambda^n +
c_{n-1}(\sigma,\tau) \lambda^{n-1}
+ \cdots \newlineeq
+c_1(\sigma,\tau)\lambda
\underbrace{\pm \overbrace{\det(A)}^{\in\KK\text{ and }\ne 0} (\tau^n+\sigma)}_{c_0(\sigma,\tau)}.
\end{multlineeq}
We shall argue that
the polynomial $c^{B(\sigma,\tau)}(\lambda)$ in (\ref{eq:cBmu}) above
is irreducible in $\KK(\sigma,\tau)[\lambda]$.
Because $B(\sigma,\tau)$ has linear forms in $\tau$ and $\sigma$ as entries,
$\deg_\tau(c_i) \le n-i$ in (\ref{eq:cBmu}).
We now suppose that
\begin{equation}\label{eq:fact}
g(\lambda,\sigma,\tau)\; h(\lambda,\sigma,\tau) = c^{B(\sigma,\tau)}(\lambda)
\end{equation}
is a non-trivial factorization in $\KK[\lambda,\sigma,\tau]$.  Then
one of the degree-$0$ coefficients in $\lambda$ of either $g$ or $h$,
which are $g(0,\sigma,\tau)$ or $h(0,\sigma,\tau)$, is a scalar
multiple of $\tau^n+\sigma$ because $c_0(\sigma,\tau) = \pm\det(A)(\tau^n+\sigma) \ne 0$
(we assumed that $A$ is non-singular) is irreducible in $\KK[\sigma,\tau]$.
Suppose $g(0,\sigma,\tau)$ is a scalar multiple of
$\tau^n+\sigma$ and consequently $h(0,\sigma,\tau) \in \KK$.
Then $\deg_\tau(g) = n$ and therefore
$\deg_\tau(h)=0$, which means by (\ref{eq:fact})
that $h$ must divide all coefficients of the powers of $\tau$
in $c^{B(\sigma,\tau)}(\lambda)$ in (\ref{eq:cBmu}).  However, the
term $\tau^n$ only occurs in $c_0(\sigma,\tau)$
and has coefficient $\pm \det(A)$ which is a non-zero field element, and
therefore $h\in\KK$, too.
Note that $c^{B(\sigma,\tau)}(\lambda)$ has
leading coefficient $1$ in $\lambda$ and therefore
no non-trivial factor in $\KK[\sigma,\tau]$ (one says it is primitive over $\KK[\sigma,\tau]$).
Then, by Gauss's Lemma, any factorization
of $c^{B(\sigma,\tau)}(\lambda)$ in $\KK(\sigma,\tau)[\lambda]$ can be
rewritten as a factorization in $\KK[\sigma,\tau][\lambda]$, and there is no no-trivial one,
so the polynomial must be irreducible also in $\KK(\sigma,\tau)[\lambda]$.
$\Box$

The Prover must convince the Verifier
that $\Delta = \det(A)$.
Here is the 2 round interactive protocol. We allow
a matrix $A$ with $\det(A)=0$ in which case the Prover
may not be able to provide the same certificate for the determinant,
but, with high probability, she cannot cheat the Verifier.
The Prover can certify $\det(A) = 0$ by a vector $w\in\KK^n$ with
$w \ne 0^n$ and $Aw = 0^n$, which the Verifier can check.
The matrix $A$ is public.

\begin{figure}[ht]
\centerline{%
\begin{tabular}{%
|%
@{\hspace*{0.1em}}r%
@{\hspace{0.25em}}l%
@{}@{}c%
@{\hspace*{0.5em}}@{}l%
@{\hspace*{0.1em}}%
|%
}%
\hline
& \multicolumn{1}{c}{{\itshape Prover}}
& \hbox to 0pt {\hss \itshape Communication \hss} &
\multicolumn{1}{c|}{{\itshape Verifier}}
\\ \hline
\rule{0pt}{4ex}%
1.&
$B = A\;\Gamma(t,s)$
& \overlongrightarrow{$t,s$}{2em} &
Checks $t^n+s\ne0$,
\\
&$t,s\in\KK$ with $t^n+s\ne 0$,
&&
\\
2.&
$c^B(\lambda) = \det(\lambda I_n - B),$
& \overlongrightarrow{$c^B$}{2em} &
\\[1ex]
3.&%
$C = [b_{i,j}]_{1\le i,j\le n-1}$,
&&
\\
&%
$c^C(\lambda) = \det(\lambda I_{n-1} - C),$
& \overlongrightarrow{$c^C$}{2em}
& Checks
\\
&$t,s\in\KK$ also with
&  &
$\quad\text{GCD}(c^B(\lambda),c^C(\lambda))=1$
\\[0ex]
&
$\text{GCD}(c^B,c^C) = 1.$
& &
\\
4.&
& \overlongleftarrow{$r_1$}{2em}
& $r_1\in S\subseteq \KK$ random
\\
& & & with $c^B(r_1) \ne 0$.
\\
5.&
Computes $w$ such that & &
\\[0ex]
& $(r_1 I_n-B)w=e_n=\begin{bmatrix}0\\[-1ex] \vdots\\ 0\\ 1\end{bmatrix}$
& \overlongrightarrow{$w$}{2em}
& Checks
\begin{tabular}[t]{@{}l@{}}$(r_1 I_n-B)w = e_n,$
\\
$w_n = c^C(r_1)/ c^B(r_1)$. %
\end{tabular}
\\
6.& & & Returns $\displaystyle \det(A) = \frac{c^B(0)}{t^n+s}$.
\\[1.5ex]%
\hline
\end{tabular}
}%
\vspace*{-2ex}%
\caption{\label{fig:noprecond}
A simple sparse determinant protocol}
\end{figure}

The interactive protocol is given in Figure~\ref{fig:noprecond}.
First we show it is  complete,
namely that if $A$ is non-singular the Prover can choose $s,t\in \KK$ such that
$\text{GCD}(c^B(\lambda),c^C(\lambda))=1$, provided $\KK$ has
sufficiently many elements.  If $A$ is singular, the Prover may not be able to do so,
in which case she can communicate that $\det(A) = 0$ and a non-zero vector $w \in \KK^n$,
$w \ne 0^n$, with $Aw = 0^n$.

Let $c^{B(\sigma,\tau)}(\lambda) = \det(\lambda I_n - A\>\Gamma(\sigma,\tau))$,
where $\Gamma(\sigma,\tau)$ is in (\ref{eq:Gamma}),
and let
$C(\sigma,\tau) = [\>(B(\sigma,\tau))_{i,j}\>]_{1\le i,j\le n-1}\in\KK[\sigma,\tau]^{(n-1)\times(n-1)}.$
We have the non-zero Sylvester resultant
\begin{equation*}
\rho(\sigma,\tau) = \operatorname{Res}_\lambda 
(c^{B(\sigma,\tau)}(\lambda),\det(\lambda I_{n-1}-C(\sigma,\tau))
) \ne 0
\end{equation*}
because any non-trivial GCD 
in $\KK(\sigma,\tau)[\lambda]$
would have to divide $c^{B(\sigma,\tau)}(\lambda)$, which is irreducible by Lemma~\ref{lem:Gamma}.
The Prover chooses $t,s$ such that $\rho(t,s)\ne 0$, which for sufficiently
large fields is possible by random selection.
Note that $c^B(\lambda)$ for such choices of $t,s$ may no longer be irreducible,
but that the resultant of $c^B$ and $c^C$ is equal to $\rho(t,s) \ne 0$ (all 
polynomials have leading coefficient $1$ in $\lambda$), hence $\text{GCD}(c^B,c^C) = 1$.

The difficulty in certificates by interaction is the proof of soundness, that
is, that the Verifier detects a dishonest Prover with high probability.
Suppose that the Prover commits
$H(\lambda) \ne c^B(\lambda)$ in place of $c^B(\lambda)$ and/or
$h(\lambda)\ne c^C(\lambda)$ in place of $c^C(\lambda)$, with
$\deg(H) = n$ and $\deg(h) = n-1$ and both
$H$ and $h$ with leading coefficient $1$ in $\lambda$.
The Prover might have chosen $t,s$ such that $\text{GCD}(c^B(\lambda), c^C(\lambda)) \ne 1$.
Or she may have been unable to compute such $t,s$ in the case when $A$ is singular,
and communicated the false $H$ and $h$ instead of presenting a linear column
relation $w$ as a certificate of singularity.
In any case, 
because $h(\lambda)/H(\lambda)$ is
a reduced fraction of polynomials with leading coefficient $1$
and because $c^B(\lambda)$ and $c^C(\lambda)$ also have leading coefficient $1$,
we must have $h(\lambda)/H(\lambda) \ne c^C(\lambda)/c^B(\lambda)$, or
equivalently $h(\lambda)c^B(\lambda)$ $-$ $H(\lambda)\*c^C(\lambda) \ne 0$.

The Verifier with probability
\begin{equation*}
\ge 1 - \frac{\deg(c^B (hc^B - Hc^C))}{|S|-n} \ge 1 - \frac{3n-2}{|S|-n}
\end{equation*}
chooses an $r_1$ such that
and $c^B(r_1) \ne 0$ and
$(h c^B - H c^C)(r_1) \ne 0$,
both of which $\Longrightarrow c^C(r_1)/c^B(r_1) \ne h(r_1)/H(r_1)$. %
The element $r_1$ satisfies $H(r_1) \ne 0$,
hence the $-n$ in the denominator, roots of $H$ are eliminated from selection set $S${}.
Note that $\deg(hc^B - Hc^C) \le 2n-2$ because the leading terms $\lambda^{2n-1}$
of both products cancel.

Suppose now that $r_1$ is chosen with those properties.  Then
\newline
$\det(r_1 I_n - B) = c^B(r_1) \ne 0$, and by Cramer's rule
\begin{equation}\label{eq:cramer}
w_n = \frac{\det(r_1 I_{n-1} - C)}{\det(r_1 I_n - B)}
    = \frac{c^C(r_1)}{c^B(r_1)}
    \ne \frac{h(r_1)}{H(r_1)}
\end{equation}
and the Verifier's last check fails.
Therefore, if the last check also succeeds, with probability
$\ge 1 - (3n-2)/(|S|-n)$ %
we have $H = c^B$ and the Verifier has the correct determinant.
Therefore the protocol is sound with high probability.

We do not fully analyze how fast the Prover could compute $t,s$, $c^B$, $c^C$,
and $w$, in a modified protocol using additional preconditioners,
as we will use the full Wiedemann technology for sparse and black box matrices in
more efficient protocols below, which we have derived from Figure~\ref{fig:noprecond}.
The Verifier in Figure~\ref{fig:noprecond} checks a polynomial GCD, chooses a random field element,
computes a matrix-times-vector product, and performs some arithmetic,
which for a sparse $A$ constitutes work more or less proportional to the input size.
Again, below we will improve the Verifier complexity, for instance
make the GCD verification $O(n)$.  The GCD = 1 property of Step~3 will
remain the fundamental ingredient for the soundness of our protocols.

\section{The Wiedemann Sequence}
\label{sec:wiedseq}

Let $A\in\KK^{n\times n}$ and $u,v\in\KK^n$.  The infinite sequence
\begin{equation}\label{eq:wiedseq}
(a_0,a_1,a_2,\swldots,a_i,\swldots)\ \text{with}\ 
a_0\sweq u^T v, %
a_i\sweq u^T A^i v\text{ for } i\swge 1,
\end{equation}
is due to D. Wiedemann~\cite{Wiedemann:1986:SSLE}.  The sequence is linearly generated
by the scalar minimal generating polynomial $f_v^{A,u}(\lambda) \in \KK[\lambda]$,
which is a factor of the minimal polynomial of the matrix $A$,
the latter of which we denote by $f^A(\lambda)$.  Both $f^A$ and $f_v^{A,u}$
are defined as monic polynomials, that is, have leading coefficient equal
to $1$.

\begin{theorem}\label{thm:fauv}
Let $S \subseteq \KK$ be of finite cardinality, which we
denote by  $|S| < \infty$, and let $u,v\in S^n$ be uniformly
randomly sampled.  Then the probability that
$f_v^{A,u} = f^A$ is at least $(1 - \deg(f^A)/|S|)^2 > 1 - 2n/|S|$
\upshape{(}cf.\ \upshape{\cite[Theorem~5]{Kaltofen:1995:ACB}}\upshape{)}.
\end{theorem}

We now define the residue for the linear generator $f_v^{A,u}$.

\begin{definition}\label{def:GAuv}
    Let $G_u^{A,v}(\lambda) = \sum_{i\geq0}a_i \lambda^{-1-i}
\in \KK[[\lambda^{-1}]]$ be the generating function of the Wiedemann
sequence\/ {\upshape(\ref{eq:wiedseq})}.
Then we define the residue $\rho_u^{A,v}(\lambda) = f_u^{A,v}(\lambda) G_u^{A,v}(\lambda)\in\KK[\lambda]$.
\end{definition}

\begin{lemma}\label{lem:rhoAuv}
The residue $\rho_u^{A,v}$ in Definition~\ref{def:GAuv} satisfies:
$$ \text{\upshape GCD}(\rho_u^{A,v}(\lambda), f_u^{A,v}(\lambda)) = 1$$
\end{lemma}
\begin{proof}
The field of quotients of the ring of power series in $\lambda^{-1}$ is denoted by
$\KK((\lambda^{-1}))$, the ring of extended power series in $\lambda^{-1}$, whose
elements can be represented as $\sum_{i\le k} c_i \lambda^i$ for $c_i \in\KK$ and $k \in \ZZ$.
The residue $\rho_u^{A,v}(\lambda)$ is computed in $\KK((\lambda^{-1}))$, but
because $f_u^{A,v}(\lambda)$ is a linear generator for $(a_i)_{i\ge 0}$, 
$\rho_u^{A,v}(\lambda)$ is a polynomial
in $\KK[\lambda]$ with $\deg(\rho_u^{A,v}) < \deg(f_u^{A,v})$.
If the greatest common divisor
$g(\lambda) = \text{GCD}(\rho_u^{A,v}(\lambda), f_u^{A,v}(\lambda))$ is
not trivial,
then the equation $\rho_u^{A,v}/g = (f_u^{A,v}/g) G_u^{A,v}$ yields
$f_u^{A,v}/g$ as a linear generator for $(a_i)_{i\ge 0}$ of lower degree than
the degree of $f_u^{A,v}$, which violates the minimality of the linear generator
$f_u^{A,v}$.
\end{proof}

\section{A Certificate for the linear \\generator} %
\label{sec:fAuvcert}

We give the 2-rounds interactive protocol  in Figure~\ref{fig:fAuvcert}.
The Prover must convince the Verifier
that $f_u^{A,v}$ is indeed the Wiedemann generator for $(u^T A^i v)_{i\ge 0}$.
The matrix $A$ and vectors $u,v$ are public.

\begin{figure}[htb]
\centerline{%
\begin{tabular}{%
|%
@{\hspace*{0.1em}}r%
@{\hspace{0.25em}}l%
@{}@{}c%
@{\hspace*{0.5em}}@{}l%
@{\hspace*{0.1em}}%
|%
}%
\hline
& \multicolumn{1}{c}{{\itshape Prover}}
& \hbox to 0em{\hss \itshape Communication \hss} &
\multicolumn{1}{c|}{{\itshape Verifier}}
\\ 
\hline
\rule{0pt}{10pt}%
1.&
$H(\lambda)=f_u^{A,v}(\lambda)$,& &\\ 	%
&$h(\lambda)=\rho_u^{A,v}(\lambda)$.&  	%
\overlongrightarrow{$H,h$}{2em}	%
&
\\[1ex]
2.&
$\phi, \psi \swin \KK[\lambda]$ with
&&
\\
&$\phi f_u^{A,v} \swplus \psi \rho_u^{A,v} \sweq 1,$
& \overlongrightarrow{$\phi,\psi$}{2em}
&
\\
&$\deg(\phi) \swle \deg(\rho_u^{A,v}) \swminus 1$,
& &
\\
&
$\deg(\psi) \swle \deg(f_u^{A,v}) \swminus 1.$
& &
\\
3.&
&
& Random $r_0 \in S\subseteq \KK$. Checks 
\\
& & &  $\text{GCD}(H(\lambda), h(\lambda))\sweq 1$ by  	%
\\
& & & $\phi(r_0) H(r_0)\swplus \psi(r_0) h(r_0)\swchecks{} 1.$
\\
&
& \overlongleftarrow{$r_1$}{2em}
& Random $r_1 \in S\subseteq \KK$.   %
\\
4.&
Computes $w$ such that
& &
\\
&$(r_1 I_n-A)w \sweq v$.
& \overlongrightarrow{$w$}{2em}
& Checks
$(r_1 I_n-A)w \swchecks{} v$
\\
&&&
and $(u^T w) H(r_1)\swchecks{}h(r_1)$.\\[1ex]	%
&&& Returns $f_u^{A,v}(\lambda)=H(\lambda)$.	%
\\[1ex] \hline
\end{tabular}
}%
\caption{\label{fig:fAuvcert}
Certificate for $f_u^{A,v}$}
\end{figure}

In Step~4 Peggy may not be able to produce a vector $w$ when
$r_1 I_n - A$ is singular.
However, she may instead convince Victor
that the random choice of $r_1$ has led to a ``failure''.
We investigate the case where the linear system is inconsistent 
more precisely. Let $f^{A,v}(\lambda)$ denote the
minimal linear generator
of the Krylov sequences of vectors $(A^i v)_{i\ge 0}$.
We have that the minimal polynomial $f^A$ of $A$ is a 
multiple of $f^{A,v}$. 
Suppose now that $f^{A,v}(r_1) = 0$ 
 and let
$\lambda_2,\ldots,\lambda_m$ be the remaining roots of $f^{A,v}$
in the algebraic closure of $\KK${}.
We obtain from $f^{A,v}(A)v = 0^n$ that
$0^n = (\prod_{j=2}^m (A - \lambda_j I_n))(A - r_1 I_n) w =
(\prod_{j=2}^m (A - \lambda_j I_n)) (-v)$, in violation
that $f^{A,v}(A) v$ constitutes the first linear dependence
of the Krylov vectors $(A^i v)_{i\ge 0}$.
Conversely,
if $f^{A,v}(r_1) \ne 0$
then the system $(r_1 I - A) w = v$ is consistent 
with  $w = (1/f^{A,v}(r_1))\,\vec\rho^{A,v}(r_1)$  (see
(\ref{eq:lambdaw}) multiplied by $\lambda I_n-A$ in the
soundness proof below).
We just proved that the linear system $(r_1 I_n - A)$ is
inconsistent with $v$ if and only if $f^{A,v}(r_1) = 0$ (hence in particular $r_1 I_n - A$ is singular).  
Peggy could provide a non-zero vector in
the right nullspace of $r_1 I_n - A$ as a proof that Victor
has communicated a bad $r_1$.  Since one can have $f_u^{A,v}(r_1) \ne 0$
when $f^{A,v}(r_1) = 0$ Victor cannot test the choice of $r_1$
before sending it even if the communicated $f_u^{A,v}$ is correct.

Given the above, the protocol 
is {\itshape perfectly
complete}: if the values $(f_u^{A,v},\rho_u^{A,v})$ and the system solution 
$w=(1/f^{A,v}(r_1))\,\vec\rho^{A,v}(r_1)$
communicated by the Prover
are correct, the Verifier always accepts $f_u^{A,v}$.
In case that the system $(r_1 I_n - A) w = v$ is inconsistent,
the Prover could communicate a Farkas certificate of inconsistency
$\bar w$ with $\bar w^T (r_1 I_n - A) = 0$ and $\bar w^T v \ne 0$.
When a Farkas certificate of inconsistency is sent,
the Verifier then 
accepts the ``correct'' output that the choice of $r_1$
has led to ``failure.'' 

In Figure~\ref{fig:fAuvcert}  we 
allow inconsistent systems to remain uncertified for two reasons:
1. A Farkas certificate requires additional work for the Prover,
possibly needing a transposed matrix times vector procedure for a black box matrix $A$. %
2. Monte Carlo algorithms are always fast and do not fail.
Their output is correct with probability $\ge$ a given bound.
Our certificates indeed have Monte Carlo randomized
verification algorithms,
that is, as we will prove below,
the interactive protocol
does not fail via system inconsistency in Step~4 and,
when all checks succeed, the
first pair of polynomials communicated is $f_u^{A,v}$ and $\rho_u^{A,v}$
with probability $\ge (1-(2n-2)/|S|)\*(1 - (3n-1)/|S|)$.

\par\vspace{\medskipamount}
\noindent
{\itshape Proof of soundness:}
Now suppose that the Prover commits
$H(\lambda)$ in place of $f_u^{A,v}(\lambda)$ and/or
$h(\lambda)$ in place of $\rho_u^{A,v}(\lambda)$, with
$n \ge \deg(H) > \deg(h)$ and $h/H \ne \rho_u^{A,v}/f_u^{A,v}$.
The Verifier with probability
\begin{equation*}
\begin{split}
&\ge 1 - \left({\deg(\phi H + \psi h)}\right)/{|S|}\\
&\ge 1 - \left( {\max\{\deg(H) + \deg(\phi), \deg(h) + \deg(\psi)\}} \right) /{|S|}\\
&\ge 1 - \left({2n-2}\right)/{|S|}
\end{split}
\end{equation*}
exposes that $(\phi H + \psi h)(r_0) \ne 1$.  Suppose now that the Verifier's check
succeeds and that actually $\text{GCD}(H,h) = 1$.
The Verifier with probability
\begin{equation*}
\begin{split}
&\ge 1 - \left( {\deg(f^A (H\rho_u^{A,v} - h f_u^{A,v}))}\right)/{|S|}\\
&\ge 1 - \left({3n-1}\right)/{|S|}
\end{split}
\end{equation*}
chooses an $r_1$ such that $f^A(r_1) \ne 0$ (that is, $r_1 I_n - A$ is non-singular) and
$(H\rho_u^{A,v} - h f_u^{A,v})(r_1) \ne 0$. The latter 
rewrites as
$h(r_1) \ne H(r_1) \rho_u^{A,v}(r_1)/f_u^{A,v}(r_1)$.
Note that $f^A(r_1) \ne 0 \Rightarrow f_u^{A,v}(r_1) \ne 0$.
Suppose now that $r_1$ is chosen with those properties.
As in the proof of Lemma~\ref{lem:rhoAuv}
we can define
\begin{align*}
    G^{A,v}(\lambda) &= \sum_{i\ge 0} \lambda^{-(i+1)} (A^i v)\in \KK^n
    [[\lambda^{-1}]],\\
    \vec\rho^{A,v}(\lambda) &= f^{A,v}(\lambda)\, G^{A,v}(\lambda) \in
    \KK^n[\lambda].
\end{align*}
We have $(\lambda I_n - A)\, G^{A,v}(\lambda) = v$ and
\begin{equation}\label{eq:lambdaw}
(\lambda I_n - A)^{-1}v = \frac{1}{f^{A,v}(\lambda)}\,\vec\rho^{A,v}(\lambda).
\end{equation}
Furthermore, $(u^T \vec\rho^{A,v})/f^{A,v} = \rho_u^{A,v}/f_u^{A,v}$,
where the left-hand side is not necessarily a reduced fraction.
Thus,
$u^T w \, H(r_1) = u^T (r_1 I_n - A)^{-1} v\, H(r_1)
= (\rho_u^{A,v}/f_u^{A,v})(r_1) \cdot H(r_1) \ne h(r_1)$,
and the Verifier's last check fails;
one can compare with (\ref{eq:cramer}), which has $u = v = e_n$ and $f_u^{B,v} = c^B${}. %
Therefore, if the last check also succeeds, with probability
$\ge (1-(2n-2)/|S|)\*(1 - (3n-1)/|S|)$ we have
$\text{GCD}(H,h)=1$ and $h/H = \rho_u^{A,v}/f_u^{A,v}$.

\par\vspace{\medskipamount}
\noindent
{\itshape About the certificate complexity:}
Excluding the input matrix $A$ and the output minimal
polynomial $f_u^{A,v}$, the certificate in Figure~\ref{fig:fAuvcert} 
comprises $\rho_u^{A,v}$, monic of degree strictly less than $n$;
then $\phi,\psi$, of degree respectively less than $n-2$ and $n-1$; and finally
$w$, a vector of $\KK^n$. The extra communications for the certificate are
thus less than $4n$ field elements.
For the time complexity, evaluating $\phi,f,\psi,\rho,f$ and $\rho$ requires less than
$6(2n)$ field operations, $Aw$ is one matrix-vector, of cost $\mu(A)$, and
checking 
$(r_1 I_n - A)w$ (either for a solution, or as a non-zero vector in the right
nullspace)  requires an additional $2n$ operations, like $u^Tw$. The
mul\-ti\-pli\-ca\-tion of the vector by $f_u^{A,v}(r_1)$ requires finally $n$ more multiplications in
the field.
We have thus proven Theorem~\ref{thm:fAuvcert} thereafter.
\begin{theorem}\label{thm:fAuvcert}
If the size of the field is $\geq 3n$,
the protocol in
Figure~\ref{fig:fAuvcert} is sound and 
complete. The associated
certificate requires less than $4n$ extra field elements and is verifiable in
less than $\mu(A)+17n$ field operations. 
\end{theorem}

Note that, with slightly larger fields, one can further reduce the complexity
for the Verifier.
\begin{corollary}\label{cor:fAuvcert}
If the size of the field is $\geq{}5n-2$, there exists a sound and perfectly
complete protocol for certifying $f_u^{A,v}$, whose associated
certificate requires less than $4n$ extra field elements and is verifiable in
less than $\mu(A)+13n$ field operations. 
\end{corollary}
\begin{proof}
For the complexity improvement, it suffices to chose $r_0=r_1$. Then the last
evaluations of $f_u^{A,v}(r_1)$ and 
$\rho_u^{A,v}(r_1)$ are already computed by the GCD check. 
Now $r_0$ must be such that 
$f^A (H\rho_u^{A,v} - h f_u^{A,v})(\phi H + \psi h-1)\,(r_0)\neq 0$, but the
latter is of degree 
$\deg(f^A)+\max\{\deg(H) + \deg(\phi), \deg(h) + \deg(\psi)\}+ \max\{\deg(H) +
\deg(\rho_u^{A,v}), \deg(h) + \deg(f_u^{A,v})\}\leq 3n-1+2n-2$.
Furthermore, the analysis in the beginning of this section shows that the
protocol of Figure~\ref{fig:fAuvcert} can be extended to be \emph{perfectly}
complete. 
\end{proof}

\begin{remark}[Certificate over small fields]\upshape
For the certificates of this paper to be sound, the underlying field cannot be
too small: $\Omega(n)$ for the minimal polynomial certificates, $\Omega(n^2)$
for the following determinant or characteristic polynomial certificates. If the
field is smaller, then a classical technique is to embed it in an extension of
adequate size.  Arithmetic operations in this
extension field $L \supset \KK$ of degree $O(\log n)$
will then cost $\log(n)^{1+o(1)}$ operations in the base field, but
with $s \ge \log_2%
(1/\epsilon)$ matrix $A$ times vectors $\in \KK^n$ operations,
that is over the base field, an incorrect $w$ is exposed with probability
$\ge 1 \swminus 1/2^s \ge 1 \swminus \epsilon$ for any constant $\epsilon > 0$. %

We give the idea for $\KK \sweq \ZZ_2$ and $L \sweq \ZZ_2[\theta]/(\chi(\theta))$
where $\chi$ is an irreducible polynomial in $\ZZ_2[\theta]$ with
$\deg(\chi) \sweq k\swplus 1 \sweq O(\log n)$.
The Prover returns $w \sweq w^{[0]} \swplus \theta w^{[1]} \swplus \swcdots \swplus \theta^{k} w^{[k]}$
with $w^{[\kappa]} \swin \ZZ_2^n$.
For
$v^{[0]} \swplus \theta v^{[1]} \swplus \swcdots \swplus \theta^{k} v^{[k]}$
$=$ $v \swminus r_1 w$
with $v^{[\kappa]}$ $\in$ $\ZZ_2^n$
we must have $-A\*w^{[\kappa]} \sweq v^{[\kappa]}$
for all $\kappa${}. %
Suppose that %
the Prover has sent a $w$ that violates $j$ of the equations,
namely, $-A w^{[\ell_i]} \sweq y^{[\ell_i]} \swne v^{[\ell_i]}$ for $1 \le i \le j$.
Then
$
|\{
(c_1,\swldots,c_j)\swin\ZZ_2^j \mid
\sum_{i=1}^j c_i y^{[\ell_i]}
\swne
\sum_{i=1}^j c_i v^{[\ell_i]}
\}| \ge 2^{j-1}.
$
The inequality %
is immediate for $j\sweq 1,2$.
For $j\ge 3$ one gets
for each
$\sum_{i=1}^{j-2} c_i y^{[\ell_i]} \sweq \sum_{i=1}^{j-2} c_i v^{[\ell_i]}$
the two unequal linear combinations for $c_{j-1} \sweq 1,c_j \sweq 0$
and $c_{j-1}\sweq 0,c_j\sweq 1$; %
one also gets for each
$\sum_{i=1}^{j-2} c_i y^{[\ell_i]} \swne \sum_{i=1}^{j-2} c_i v^{[\ell_i]}$
the unequal linear combination for $c_{j-1} \sweq c_j \sweq 0$;
a second unequal combination is either for $c_{j-1}$ $=$ $1$, $c_j$ $=$ $0$
or $c_{j-1} \sweq c_j \sweq 1$, because not both can be equal at the same time.
Therefore
$
|\{
(c_0,\swldots,c_k)\swin\ZZ_2^{k+1} \mid
-A
(\sum_{i=0}^k c_i w^{[i]})
$
$\ne$
$
(\sum_{i=0}^k c_i\*v^{[i]})
$
$\}|$ $\ge$ $2^k $
which shows that checking a random linear combination over $\ZZ_2$
exposes an incorrect~$w$
with probability $\ge 1/2$.
The cost for the Verifier %
is thus $O(\mu(A))+n^{1+o(1)}$ operations in $\ZZ_2$.
\end{remark}

\par\vspace{\medskipamount}
\noindent
{\itshape Complexity for the Prover:} 
Theorem~\ref{thm:fAuvcert} places no requirement on how the Prover
computes the required commitments. Coppersmith's block Wiedemann method~\cite{Coppersmith:1994:SHL},
for instance, can be
used.
\begin{theorem}
If the size of the field is $>3n$ then the Prover can produce the certificate in Figure~\ref{fig:fAuvcert} in no more
than $(1+o(1))n$ multiplications of $A$ times a vector in $\KK^n$ and   $n^{2+o(1)}$ additional  operations in $\KK$,
using $n^{1+o(1)}$ auxiliary storage for elements in $\KK$. %
\end{theorem}
\begin{proof}
The Berlekamp-Massey algorithm can be applied on the
sequence $(a_i) _{i\geq 0} =(u^TA^iv)_{i\geq 0}$ to recover $f_u^{A,v}$. However that requires (at
most) $2n$ terms to be
computed~\cite{Wiedemann:1986:SSLE}, and hence more multiplications of
$A$ times a vector than requested in the statement. 
Following~\cite[Sec.\,7]{Kaltofen:1995:ACB}, the number of applications of $A$ can be minimized using 
Coppersmith's block Wiedemann algorithm~\cite{Coppersmith:1994:SHL} with a special choice of parameters. 
We consider integers $p$ and $q$ (the blocking factors) such that
$q=o(p)$, $p=n^{o(1)}$, and $1/q=o(1)$. 
For example, we may
choose $p=(\log n)^2$ and $q=\log n$. 

For the block Wiedemann algorithm we first take random $\mathcal{U}_0\in \KK^{n \times (p-1)}$ and 
 $\mathcal{V}_0\in \KK^{n \times (q-1)}$ with entries sampled from $S\subseteq \KK$, 
and construct $\mathcal{U}=[u,\mathcal{U}_0]\in \KK^{n \times p}$ and 
$\mathcal{V}=[v,\mathcal{V}_0]\in \KK^{n \times q}$. 
With $d_l=\lceil {n}/{(p-1)}\rceil$ and $d_r=\lceil  {n}/(q-1)\rceil$, 
the rank of the  block Hankel matrix 
$\text{Hk}_{d_l,d_r}(A,\mathcal{U}_0,\mathcal{V}_0)$ 
in~\cite[(2.4)]{KaVi05}
is maximal with high probability, and equal to the dimension of the Krylov space $\mathcal{K}_{A,\mathcal{V}_0}$ (using $q \leq p$).
Since the degree of the $q$-th (resp. the $p$-th) highest degree invariant factor of $\lambda -A$ has degree less than 
 $d_r$ (resp. $d_l$),  less than $d_r$  (resp. $d_l$)  first Krylov iterates $(A^iv)_{i\geq 0}$ (resp. $(u^TA^i)_{i\geq 0}$) suffice 
for completing any basis of  $\mathcal{K}_{A,\mathcal{V}_0}$  into a basis of  $\mathcal{K}_{A,\mathcal{V}}$.  
(resp. $\mathcal{K}_{A,\mathcal{U}_0}$ and $\mathcal{K}_{A,\mathcal{U}}$).  %
Therefore the rank of $\text{Hk}_{d_l,d_r}(A,\mathcal{U},\mathcal{V})$ must be maximal and equal to 
$\dim  \mathcal{K}_{A,\mathcal{V}}$. 

The determinant of a submatrix of 
$\text{Hk}_{d_l,d_r}(A,\mathcal{U},\mathcal{V})$ with maximal rank 
has degree at most $2n$ in the entries of 
$\mathcal{U}_0$ and $\mathcal{V}_0$. Using arguments analogous to those in \cite[sec.\,9.2]{Villard:1997:further} 
or \cite[Sec.\,3.2]{KaVi05}, 
by the DeMillo-Lipton/\allowbreak Schwartz/\allowbreak Zippel Lemma%
~\cite{DeMLip78,Zippel:1979:ZSlemma,Schwartz:1980:SZlemma},
we obtain that:
\begin{equation} \label{eq:hankel}
\text{Prob} \left( \text{\rm rank}\,\text{Hk}_{d_l,d_r}(A,\mathcal{U},\mathcal{V}) = \dim  \mathcal{K}_{A,\mathcal{V}}\right) \geq 1-2n/|S|. 
\end{equation} 
If $S$ contains sufficiently many elements, it follows~\cite{Villard:1997:further,KaVi05}  that it is sufficient to compute the  
sequence of $p\times q$ matrices $\mathcal{A}_i=\mathcal{U}^TA^i\mathcal{V}$
up to $L=d_l +d_r$ terms, 
 which can
be achieved with $qL=n(1+o(1))$ multiplications of $A$ times a vector, and $n^{2+o(1)}$ additional operations in $\KK$. Let
$$\mathcal{G}_{\mathcal{U}}^{A,\mathcal{V}}(\lambda)=
\sum_{i\geq0}\mathcal{A}_i \lambda^{-1-i} = 
\sum_{i\geq0}\mathcal{U}^TA^i\mathcal{V}\lambda^{-1-i}\in\KK^{p\times
q}[[\lambda^{-1}]]$$ be the generating function for the block Wiedemann
sequence.  
One can then compute 
a rational fraction description of
$\mathcal{G}_{\mathcal{U}}^{A,\mathcal{V}}$, in the form of a pair of
matrices $\mathcal{F}(\lambda) \in\KK^{q\times
q}[\lambda]$ and $\mathcal{R}(\lambda) \in\KK^{p\times
q}[\lambda]$ with degree $d_r=o(n)$, so that $\mathcal{F}$ is a minimal matrix generator for $(\mathcal{A}_i)_{i\geq 0}$, and 
$$\mathcal{G}_{\mathcal{U}}^{A,\mathcal{V}}(\lambda)\mathcal{F}(\lambda)=\mathcal{R}(\lambda) \in\KK^{p\times
q}[\lambda].$$
The matrices  $\mathcal{F}$ and  $\mathcal{R}$ can be computed in time
soft-%
linear in $n$~\cite{BeLa94}, 
and since the dimensions are in $n^{o(1)}$, the matrix fraction
$\mathcal{R}(\lambda)\*\mathcal{F}(\lambda)^{-1}$ is also obtained in 
soft-%
linear time. Any of the entries of this matrix fraction is a rational fraction
description of a sequence corresponding to two chosen vectors among the
blocking vectors. In particular, from the entry (1,1), Peggy obtains the rational fraction
description $\rho_u^{A,v}/f_u^{A,v}$ of $(a_i)_{i\geq0}$.
The computation of $\phi$ and $\psi$ follows via an extended Euclidean
algorithm applied on $f_u^{A,v}$ and $\rho_u^{A,v}$.

A small variation of the above scheme solves the inhomogenous linear
system in Step 4.  Let $B=r_1 I_n - A$.  With high probability, Victor
chooses a $r_1$ such that $f^A(r_1)\neq 0$, hence $B$ is non-singular.
Together with~\eqref{eq:hankel} this gives the requirement $3n+1$ for the
field size.

We have that ${\mathcal F}'={\mathcal F}(r_1-\lambda)$ is a minimal
matrix generator for the sequence of $p\times q$ matrices $(\mathcal{U}^T
B^i\mathcal{V})_{i\geq 0}$ (computing ${\mathcal F}'$ is a soft-linear
operation).  Since $B$ is non-singular we have $\det {\mathcal F}'(0)\neq
0$. With column operations on ${\mathcal F}'$, we can arrange so that one
column becomes $[f_0,0,\ldots,0]^T$ with $f_0 \neq 0$. $\mathcal{V}$~times the
latter being $f_0.v$, following the
inhomogeneous case in \cite[Sec.\,8]{Coppersmith:1994:SHL}, this suffices
to solve the linear system $Bw=v$, using $d_r=o(n)$ more matrix-vector
multiplications and $n^{2+o(1)}$ operations in~$\KK$.
\end{proof}

\section{A Certificate for the minimal polynomial}\label{sec:minpol}
With random vectors  
$u$ and $v$, we address the certification of the minimal polynomial
$f^A$. 
Indeed, 
using Wiedemann's study~\cite[Sec.\,VI]{Wiedemann:1986:SSLE} or the alternative
approach in~\cite{KaPa91,Kaltofen:1995:ACB}, we know that  
for random $u$ and $v$, $f_u^{A,v}$ is equal to $f^A$ with high
probability (see Theorem~\ref{thm:fauv}). 
This is shown in Figure~\ref{fig:MinPoly}.

\begin{figure}[ht]
\centerline{%
\framebox{%
    \begin{tabular}{@{}r@{\hspace{0em}}l@{}@{}c@{}c@{}cl@{}}
& \multicolumn{1}{c}{
        \hbox to 1em {\hss \itshape\hskip 0em
        Prover \hss}
}
    & \multicolumn{3}{c}{\hbox to 0pt {\hss \itshape\hskip 2em
Communication \hss}} &
\multicolumn{1}{c}{{\itshape Verifier}} 
\\ \hline
\rule{0pt}{4ex}%
1.&&&\overlongleftarrow{$u,v$}{5em} &&Random $u,v\in S^n\subseteq \KK^{n}$. %
\\[2ex]
\cline{3-3}
\cline{5-5}
&&\multicolumn{1}{|c}{\hskip 0.25em}&\smash{\raisebox{1ex}{Figure~\ref{fig:fAuvcert}}}&
\multicolumn{1}{c|}{\hskip 0.25em}\\
2.&&\multicolumn{3}{|c|}{\overlongrightarrow{$H,h,\phi,\psi$}{5em}}&\\
3.&&\multicolumn{3}{|c|}{\overlongleftarrow{$r_1$}{5em}}&\\
4.&&\multicolumn{3}{|c|}{\overlongrightarrow{$w$}{5em}}&{Checks
  $H\checks{}f_u^{A,v}$, w.h.p.}\\
\cline{3-5}
\\
5.& && & & {Returns $f^A=f_u^{A,v}$, w.h.p.}
\end{tabular}
}%
}%
\caption{\label{fig:MinPoly}
Certificate for $f^A$ with random projections}
\end{figure}

\begin{proposition} The protocol in Figure~\ref{fig:MinPoly} is sound and
  complete.
\end{proposition}
\begin{proof}
First, the protocol is complete. Indeed,  the result is correct means that
$H=f^A$. 
Thus, if $f^A=f_u^{A,v}$, then completeness is guaranteed by the completeness in
Theorem~\ref{thm:fAuvcert}. 
Otherwise, $H$ is a proper multiple of $f_u^{A,v}$ and $H=f^A\neq f_u^{A,v}$.
If $u,v$ are randomly chosen by Victor,
this happens only with low probability, thanks to Theorem~\ref{thm:fauv}.

Second, for the soundness, if the result is incorrect, then $H\neq{}f^A$.
If also $H\neq f_u^{A,v}$, then similarly, $H$ will make the certificate of
Figure~\ref{fig:fAuvcert} fail with high probability, by the soundness in
Theorem~\ref{thm:fAuvcert}.
Otherwise, $H=f_u^{A,v}\neq{}f^A$, but if $u,v$ are randomly chosen by Victor,
this happens only with low probability, thanks to Theorem~\ref{thm:fauv}.
\end{proof}

\begin{corollary}\label{cor:MinPoly}
If the size of the field is $\geq{}5n-2$, 
there exists a sound and perfectly
complete protocol for certifying $f^A$, whose associated
certificate requires 
less than $8n$ extra field elements and is verifiable 
in less than $2\mu(A)+26n$ field operations.
\end{corollary}
\begin{proof}
We add some work for the Prover, and double the certificate.
First Peggy has to detect that the projections given by Victor reveal only a
proper factor of the minimal polynomial. Second she needs to prove to Victor
that he was wrong. For this, Peggy:
{\setlength{\leftmargini}{\wd\bulspace}%
\begin{enumerate}
\item[$\bullet$]
Computes the minimal polynomial $f^A$ of the matrix and the minimal
  polynomial $f_u^{A,v}$.  \\[-0.5cm]
\item[$\bullet$]
If $f^A \neq f_u^{A,v}$, Peggy searches for projections
  $(\hat{u},\hat{v})$, such that $\deg(f_{\hat{u}}^{A,\hat{v}}) >
  \deg(f_u^{A,v})$.\\[-0.5cm]
\item[$\bullet$]
Peggy then starts two certificates of Figure~\ref{fig:fAuvcert}, one for
  $(u,v)$, one for $(\hat{u},\hat{v})$.
\end{enumerate}
}%
In case of success of the latter two certificates, Victor is convinced that 
both $f_u^{A,v}$ and $f_{\hat{u}}^{A,\hat{v}}$ are correct.
But as the latter polynomial has a degree strictly larger than the former, he is
also convinced that his projections can not reveal $f^A$.
Complexities are given by applying Corollary~\ref{cor:fAuvcert} twice.
\end{proof}

\begin{remark}[Certificate for the rank]\upshape %
If $A$ is non-sin\-gu\-lar, the certificate
of~{\upshape\cite[Figure~2]{jgd:2014:interactivecert}} can be used to certify non-singularity:
for any random vector $b$ proposed by Victor, Peggy can solve the
system $Aw=b$ and return $w$ as a certificate.
Now, if $A$ is singular, a similar idea %
as for the minimal polynomial
can be used to certify the rank: precondition the matrix $A$
into a modified matrix $B$ whose
minimal polynomial is $f(x)x$ and characteristic polynomial is
$f(x)x^k$, where $f(0)\neq{}0$. As a consequence $\rank(A)=n-k$.
For instance, if $A$ is symmetric, such a
{\sc PreCondCycNil} preconditioner can be a non-singular diagonal
matrix $D$ if the field is sufficiently
large~{\upshape\cite[Theorem~4.3]{CEKSTV02}}.
Otherwise, $A^T D_2 A$ is symmetric, with $D_2$ another diagonal
matrix. Then the minimal polynomial certificate can be applied to
$B=DA^TD_2AD$~{\upshape\cite{EK97}}. 
Comparing with the certificate for the rank
in~{\upshape\cite[Corollary~3]{jgd:2014:interactivecert}}, 
this new certificate saves a logarithmic factor
in the verification time, but requires the field to be larger (from
$\Omega(n)$ to $\Omega(n^2)$).
\end{remark}

\section{Certificate for the determinant with Diagonal
  preconditioning}\label{sec:diag}

First of all, if $A$ is singular, Peggy may not be able to produce the
desired certificate. In which case she can communicate that $\det(A) =
0$ and produce a non-zero vector in the kernel: 
$w\in\KK^n$, $w\neq{}0^n$, with $Aw = 0^n$.

We thus suppose in the following that $A$ is non-singular.

The idea of~\cite[Theorem~2]{Wiedemann:1986:SSLE} is to precondition the
initial invertible matrix $A$ into a modified matrix $B$ whose
characteristic polynomial is square-free, and whose determinant is an
easily computable modification of that of $A$. For instance, such a
{\sc PreCondCyc} preconditioner can be a non-singular diagonal
matrix $D$ if the field is sufficiently large~\cite[Theorem~4.2]{CEKSTV02}:
\begin{equation}\label{eq:probDA}
\text{Prob} \left(\deg\left(f^{DA}\right)=n\right) \geq 1 - \frac{n(n-1)}{2|S|} 
\end{equation}
To certify the determinant, it is thus sufficient for Peggy to chose
a non-singular diagonal matrix $D$ and two vectors $u,v$ such that 
$\deg(f_u^{DA,v})=n$ and then to use the minimal polynomial
certificate for $f_u^{DA,v}$, as shown on Figure~\ref{fig:WiedDet}.

\begin{figure}[ht]
\centerline{%
\framebox{%
    \begin{tabular}{@{}r@{\hspace{0.25em}}l@{}@{}c@{}c@{}cl@{}}
& \multicolumn{1}{c}{{\itshape Prover}}
        && \hbox to 0pt {\hss \itshape Communication \hss} &&
\multicolumn{1}{c}{{\itshape Verifier}}
\\ \hline
\rule{0pt}{4ex}%
1.&Form $B=DA$ with&&\\
  &$D\in S^{n}\subseteq \KK^{*n}$&&\overlongrightarrow{$D,u,v$}{5em} &&\\
  &and $u,v\in S^n$, &&\\
  &s.t. $\deg(f_u^{B,v})=n$. &&\\
\cline{3-3}
\cline{5-5}
&&\multicolumn{1}{|c}{\hskip 0.25em}&\smash{\raisebox{1ex}{Figure~\ref{fig:fAuvcert}}}&
\multicolumn{1}{c|}{\hskip 0.25em}\\
2.&&\multicolumn{3}{|c|}{\overlongrightarrow{$H,h,\phi,\psi$}{5em}}&{Checks:}\\
3.&&\multicolumn{3}{|c|}{\overlongleftarrow{$r_1$}{5em}}&{$\deg(H)\checks{}n$,}\\
4.&&\multicolumn{3}{|c|}{\overlongrightarrow{$w$}{5em}}&{$H\checks{}f_u^{B,v}$, w.h.p.}\\
\cline{3-5}
\\
5.&&&&& {Returns $\displaystyle\frac{f_u^{B,v}(0)}{\det(D)}$.}
\end{tabular}
}%
}%
\caption{\label{fig:WiedDet}
Determinant certificate with Diagonal preconditioning for a
non-singular matrix}
\end{figure}

\begin{remark}\upshape
Note that in the minimal polynomial sub-routine of Figure~\ref{fig:WiedDet},
Peggy can actually choose $D,u,v$, since the check on the degree of
$f_u^{B,v}$ prevents bad choices for $D,u,v$. Victor could also 
select them himself, the overall volume of communications
would be unchanged, but he would have to perform more work, namely
selecting $3n+2$ random elements instead of just~$1$.
\end{remark}

\begin{theorem}\label{thm:WiedDet}
If the size of the field is
$\geq\max\{\frac{1}{2}n^2-\frac{1}{2}n,5n-2\}$, the protocol for the 
determinant of a non-singular matrix in Figure~\ref{fig:WiedDet} is
sound and complete. 
The associated certificate requires less than $8n$ extra field
elements and is verifiable in less than $\mu(A)+15n$ field operations.
If the size of the field is
$\geq{}n^2+n+5$, with high probability the Prover can produce it with
no more than $2$ minimal polynomial computations, %
 and $1$ system solving.
\end{theorem}
\begin{proof}
First, if the Prover is honest, $H$ is correctly checked to be
$f_u^{B,v}$, as the minimal polynomial is complete. 
Then, by definition, $f_u^{B,v}$ is a factor of the characteristic polynomial of
$B$. But if its degree is $n$, then it is the characteristic
polynomial. Therefore its unit coefficient is the determinant of $B$
and the certificate of Figure~\ref{fig:WiedDet} is complete.

Second, for the soundness. If $H\neq{}f_u^{B,v}$, then the minimal
polynomial certificate will most probably fail, by the soundness of the minimal
polynomial certificate, as $|S|\geq{}(5n-2)$, from Corollary~\ref{cor:fAuvcert}.
Otherwise, the degree check enforces that
$f_u^{B,v}$ is the characteristic polynomial.

Now, for the complexity, 
with respect to the minimal polynomial certificate and
Theorem~\ref{thm:fAuvcert}, this certificate requires an extra
diagonal matrix $D$. As $u,v,f_u^{B,v}$ are not input/output anymore,
the extra communications grow from $4n$ field elements to $8n$.
The verification procedure is similar except that verifying a linear system
solution  
with $B$ requires $n+\mu(A)$ operations %
and that $det(D)$ has to be
computed, hence the supplementary $2n$ field operations.
Finally, for the Prover,
in order to find suitable vectors and diagonal matrices, 
Peggy can select them randomly in $S^n$ and try $f_u^{DA,v}$
until $\deg(f_u^{DA,v})=n$. It will succeed with the joint
probabilities of Theorem~\ref{thm:fauv} and Equation~(\ref{eq:probDA}).
For $n\geq{}2$, as soon as
$|S|\geq{}n^2+n+5$,
the lower bound of the %
probability of success 
$(1 - \relax{n(n-1)}/{(2|S|)})(1 - \relax{2n}/{|S|})$
is higher than~$\relax{1}/{2}$ and the expected number of trials for the
Prover is less than~$2$. 
\end{proof}

Note that the
certificate (as well as that of next section) can be made perfectly complete, by using the perfectly
complete certificate for the minimal polynomial of
Corollary~\ref{cor:MinPoly}.

\section{Determinant with Gamma\\ pre\-con\-di\-tioning}\label{sec:gammadet}
In order to compute the determinant via a minimal polynomial,
the diagonal preconditioning ensures that it is equal to the characteristic
polynomial, as the latter is square-free with factors of degree $1$.  
With the preconditioning of Section~\ref{sec:simplecert}, we can differently
ensure that the characteristic polynomial is irreducible.
It has the same effect, that it equals the minimal polynomial, but it
also enforces that it has no smaller degree factors. Therefore, for any
non-zero $u$ and $v$, as a non-singular matrix is non-zero, the only
possibility for $f_u^{A,v}$ is to be of degree $n$. 
Thus, we can give in Figure~\ref{fig:GammaDet} an improved
certificate. It chooses pre-determined vectors $u$ and $v$ to be as
simple as possible: the canonical vector $e_1=[1,0,\ldots,0]^T$.

\begin{figure}[ht]
\centerline{%
\framebox{
    \begin{tabular}{@{}r@{\hspace{0.25em}}l@{}@{}c@{}c@{}cl@{}}
& \multicolumn{1}{c}{{\itshape Prover}}
        && \hbox to 0pt {\hss \itshape Communication \hss} &&
\multicolumn{1}{c}{{\itshape Verifier}}
\\ \hline
\rule{0pt}{4ex}%
1.&Form $B=A\>\Gamma(s,t)$&&\\
  &with $s,t\in\KK$ s.t. && \overlongrightarrow{$s,t$}{5em} && Checks $t^n+s\ne0$.\\
  &$\Gamma(s,t)$ non-singular && &&\\
  &($\Leftrightarrow$ $t^n+s\ne 0$) &&\\
  & and $H_n(s,t)\neq 0$. && \\
  & (see~(\ref{eq:Gamma}) and (\ref{eq:Hke})) && \\
  \cline{3-3}
  \cline{5-5}
&&\multicolumn{1}{|c}{\hskip 0.25em}&\smash{\raisebox{1ex}{Figure~\ref{fig:fAuvcert}}}&
\multicolumn{1}{c|}{\hskip 0.25em}\\
2.&&\multicolumn{3}{|c|}{\overlongrightarrow{$H,h,\phi,\psi$}{5em}}&Checks:\\
3.&&\multicolumn{3}{|c|}{\overlongleftarrow{$r_1$}{5em}}&$\deg(H)\checks{}n$,\\
4.&&\multicolumn{3}{|c|}{\overlongrightarrow{$w$}{5em}}&$H\checks{}f_{e_1}^{B,e_1}$, w.h.p.\\
\cline{3-5}
\\
5.&&&&& {Returns $\displaystyle\frac{f_{e_1}^{B,e_1}(0)}{t^n+s}$.}
\end{tabular}
}%
}%
\caption{\label{fig:GammaDet}
Determinant certificate with Gamma preconditioning for a
non-singular matrix}
\end{figure}

\begin{theorem} 
If the size of the field is $\geq\max\{n^2-n,5n-2\}$, the protocol for the
determinant of a non-singular matrix in Figure~\ref{fig:GammaDet} is
sound and complete. 
The associated certificate requires less than $5n$ extra field
elements and is verifiable in less than $\mu(A)+13n+o(n)$ field operations.
If the size of the field is $\geq{}2n^2-2n$, with high probability the
Prover can produce it with no more than $2$ minimal polynomial
computations, and $1$ system solving. 
\end{theorem}
\noindent{\sc Proof.}
Completeness is given by the same argument as for
Theorem~\ref{thm:WiedDet}.
It is similar for the soundness, provided that $H_n(s,t)\neq 0$
implies that $\deg(f_{e_1}^{B,e_1})=n$.

Therefore, let $C^{B(\sigma,\tau)}(\lambda) = \det(\lambda I_n - A\>\Gamma(\sigma,\tau))$,
where $\Gamma(\sigma,\tau)$ is in (\ref{eq:Gamma}).
Then, for $e_1=[1,0,\ldots,0]^T$ the first canonical vector,
let $G_1^{B(\sigma,\tau)}(\lambda) = e_1^T(\lambda I_n -
A)^{-1}e_{1}$, and 
let $\rho_1^{(\sigma,\tau)}(\lambda) = {C^{B(\sigma,\tau)}(\lambda) }G_1^{B(\sigma,\tau)}(\lambda) .$

As $C^{B(\sigma,\tau)}$ is monic and $e_1^Te_1=1$, then $G_1^{B(\sigma,\tau)}$
and $\rho_1^{(\sigma,\tau)}$ are not identically zero.
Further, $C^{B(\sigma,\tau)}$ is irreducible by Lemma~\ref{lem:Gamma}
but the minimal polynomial for the projected sequence must divide
$C^B$.
Therefore the sequence has $C^B$ for minimal polynomial.
Now, this sequence is denoted by $(s_i(\sigma,\tau))=(e_1^T B(\sigma,\tau)^i e_1)$, 
and let $H_n(\sigma,\tau) \in\KK(\sigma,\tau)^{n\times n}$:
\begin{equation}\label{eq:Hke}\small
H_n(\sigma,\tau)=
\left[\begin{matrix} 
s_0(\sigma,\tau) & s_1(\sigma,\tau) & \ldots & s_{n-1}(\sigma,\tau) \\
s_1(\sigma,\tau) & s_2(\sigma,\tau) & \ldots & s_{n}(\sigma,\tau)   \\
\vdots& & \ddots & \vdots    \\
s_{n-1}(\sigma,\tau)& s_n(\sigma,\tau) & \ldots & s_{2n-2}(\sigma,\tau)
\end{matrix}\right].\end{equation}
Since the minimal polynomial of the sequence is of degree $n$,
$H_n(\sigma$, %
$\tau)$ is non-singular~\cite[Eq.~(2.6)]{KaVi05}.
Thus, any $s,t$ with $\det( H_n(s$, %
$t) ) \ne 0$ will yield 
$f_{e_1}^{B(s,t),e_1} = C^{B(s,t)}$. 
Now for the complexities:
{\setlength{\leftmargini}{\wd\bulspace}%
\begin{enumerate}
\item[$\bullet$]
The volume of communication is reduced, from~$8n$ to~$5n$.\\[-0.5cm]
\item[$\bullet$]
With respect to Theorem~\ref{thm:WiedDet}, the cost for Victor
  is slightly improved: an application by $D$ replaced by an
  application by $\Gamma(s,t)$, a dot-product with $u$ replaced by
  one with $e_1$, and $n$ 
  operations for $\det(D)$ are replaced by less than
  $2\lceil\log_2(n)\rceil+1$ to compute $t^n+s$. The overall
  verification cost thus decreases from 
  $\mu(A)+15n+o(n)$ to $\mu(A)+13n+o(n)$; \\[-0.5cm]
\item[$\bullet$]
For Peggy, applying the diagonal scaling costs $n$ per
  iteration, while applying $\Gamma(s,t)$ costs $2n$ operations per
  iteration; but applying a random $u^T$ costs about $2n$ operations
  per iteration, while applying $e_1^T$ is just selecting the first
  coefficient, so her cost is slightly improved. 
  Then, too choose such $s,t$, Peggy can try uniformly sampled
  elements in $S$, and see whether
  $\deg(f_{e_1}^{B(s,t),e_1})=n$. Since $\deg(\Gamma(\sigma,\tau))=1$,
  we have that $\deg(s_i(\sigma,\tau))\leq{}i$ and 
  $\deg(\det( H_n(\sigma,\tau) ))\leq{}n(n-1)$.  
  Hence, by the
DeMillo-Lipton/Schwartz/Zippel Lemma %
  \begin{equation}
    \text{Prob}\left(\det( H_n(s,t) ) \ne 0\right)\geq
    1-\frac{n(n-1)}{|S|}.
  \end{equation}  
  As soon as $|S|\geq{}2n(n-1)$, the probability of success 
for  Peggy 
is thus larger than $1/2$ %
and the expected number of trials
  is less than~$2$.\,$\Box$
\end{enumerate}
}%

We gather the differences between the protocols of
Figure~\ref{fig:WiedDet} and~\ref{fig:GammaDet} in
Table~\ref{tab:recap}. 
The two
certificate differ mostly only in the preconditioning.
But this allows to gain a lot of randomization:
the number of random field elements per try to sample for Peggy is
reduced from $3n$ to only~$2$. As the size of the set is $\Omega(n^2)$, 
this reduces the number of random bits from $O(n\log(n))$ to $O(\log(n))$.

\begin{table}[htb]
\renewcommand{\arraystretch}{1.3}
\[\begin{array}{|c||c|c|}
\hline
\multicolumn{3}{|c|}{\text{Certificates for the determinant of sparse matrices}}  \\
\hline
\hline
{\text{Preconditioner}}& \S~\ref{sec:diag}: D  & \S~\ref{sec:gammadet}: \Gamma(t,s) \\
\hline
{\text{Verifier}} & \mu(A)+15n+o(n) & \mu(A)+13n+o(n)\\
{\text{Communications}} & 8n & 5n\\
{\text{Random elements}} & 3n+2 & 3\\
\hline
{\text{Prover}} & 
\multicolumn{2}{|c|}{\text{\sc MinPoly}(n)+\text{\sc LinSys}(n)}\\
\hline
{\text{Field size}} & \geq{}\frac{1}{2}(n^2-n) &  \geq{}n^2-n\\
\hline
\end{array}\] %
\caption{\mdseries %
Summary of the complexity bounds of the certificates presented in this
  paper for the determinant of sparse matrices 
  ($n$ is the dimension, %
$\mu(A)$  bounds the cost of one
  matrix-vector product, {\sc MinPoly}$(n)$ (resp. {\sc LinSys}$(n)$) is 
  the cost of computing the minimal polynomial of a sequence %
(resp. of
  solving a linear system).}\label{tab:recap} 
\end{table}
\medskip

\begin{remark}[Certificate for the characteristic polynomial]\upshape%
The certificate for the determinant can be combined with
the characteristic polynomial reduction
of~\cite[Figure~1]{jgd:2014:interactivecert}.
As this reduction is linear for Victor, this then provides now also a
{\em linear} time verification procedure for the characteristic
polynomial:
\begin{enumerate}
\item Peggy sends $c^A$ as the characteristic polynomial; \\[-0.5cm]
\item Victor now sends back a random point $r\in{}S\subseteq\KK$; \\[-0.5cm]
\item Peggy and Victor enter a determinant certificate for $r I-A$; \\[-0.5cm]
\item Once convinced, Victor checks that $\text{\rm det}(r I-A)\checks{}c^A(r)$.
\end{enumerate}
For a 
random 
$r\in{}S$, in the determinant
sub-certificate, $r I-A$ will be non-singular if
$c^A(r)\neq{}0$, hence with probability 
$\geq 1-{n}/{|S|}$. Then $\text{\rm det}(r I-A)$  is certified
using the certificate of Figure~\ref{fig:GammaDet}.

Best known algorithms for computing the characteristic polynomial, using 
either quadratic space and fast matrix multiplication or linear space~\cite{Vil00}, 
have cost $O(n^{2+\alpha})$ for some $\alpha >0$.  
The characteristic polynomial  with $\mu(A)=O(n)$ is thus an example of a problem 
whose worst-case complexity bound is super-quadratic in the verification cost $O(n)$.
\end{remark}
\bibliographystyle{abbrvurl}
\input{thebib}
\bibliography{zkfhe}
\end{document}